\newif\ifabstract
\newif\iffull\fulltrue

\abstracttrue\fullfalse

\documentclass{llncs}

\usepackage[numbers]{natbib}
\usepackage{amsfonts}
\usepackage{amssymb}
\usepackage{amsmath}
\usepackage{algorithmic}
\usepackage{algorithm}
\usepackage{color}
\usepackage{xcolor}
\usepackage{tikz}
\usepackage{wrapfig}
\usepackage{comment}
\usepackage{scalefnt}

\makeatletter
\renewcommand\@biblabel[1]{#1.}
\makeatother

\title{Funding Games: \\the Truth but not the Whole Truth} 

\author{Amotz Bar-Noy \inst{1} \and Yi Gai \inst{2} \and Matthew P. Johnson \inst{3} \and Bhaskar Krishnamachari \inst{2} \and George Rabanca \inst{1} }
\institute{Department of Computer Science, Graduate Center, City University of New York\\
\email{amotz@sci.brooklyn.cuny.edu}, \email{grabanca@gc.cuny.edu}
\and Ming Hsieh Department of Electrical Engineering, University of Southern California\\
\email{ygai@usc.edu}, \email{bkrishna@usc.edu}
\and Department of Electrical Engineering, University of California\\
\email{mpjohnson@gmail.com}
}

\date{}

\begin{document}

\maketitle

\begin{abstract}
We introduce the Funding Game, in which $m$ identical resources are to be allocated among $n$ selfish agents. Each agent requests a number of resources $x_i$ and reports a valuation $\tilde{v}_i(x_i)$, which verifiably {\em lower}-bounds $i$'s true value for receiving $x_i$ items.  The pairs $(x_i, \tilde{v}_i(x_i))$ can be thought of as size-value pairs defining a knapsack problem with capacity $m$. A publicly-known algorithm is used to solve this knapsack problem, deciding which requests to satisfy in order to maximize the social welfare.

We show that a simple mechanism based on the knapsack {\it highest ratio greedy} algorithm provides a Bayesian Price of Anarchy of 2, and for the complete information version of the game we give an algorithm that computes a Nash equilibrium strategy profile in $O(n^2 \log^2 m)$ time.  Our primary algorithmic result shows that an extension of the mechanism to $k$ rounds has a Price of Anarchy of $1 + \frac{1}{k}$, yielding a graceful tradeoff between communication complexity and the social welfare.
\end{abstract}

\section{Introduction}

Efficiently allocating resources among multiple potential recipients is a central problem in both computer science and economics.  In the mechanism design literature it is customary to use the power of currency exchange to provide incentives for the agents to be truthful.  However, it has been pointed out that assuming the existence of currency in the model is not always justified (\cite{Procaccia09}).  In the present paper we initiate the study of mechanisms with verification, first introduced by Nisan and Ronen in \cite{Nisan01} for the job scheduling problem, for resource allocation problems in a setting without currency.  This reveals an unexplored middle ground area between the settings of the multiple choice knapsack problem and that of multi-item auctions, which has some obvious practical applications.

The knapsack problem and its variations model the setting where the supplier knows precisely what value the agents are getting from any number of items.  This can be thought of as a perfect verification mechanism, and selfishness does not play a role.  At the other extreme, work in algorithmic game theory has generally considered the case where the supplier knows nothing about the agents' valuation and must provide incentives, typically by imposing payments, for the agents to be truthful.   
In this paper we introduce the Funding Game, in which a supplier distributes $m$ identical resources among $n$ agents, each of whom has a private valuation function depending only on the number of items received.
Each agent requests a number of items $x_i$, and specifies its value $\tilde{v}_i(x_i)$ for these items, which might be less that its real value $v_i(x_i)$.  The supplier can verify that the valuations are not exaggerated, and uses a publicly known algorithm to allocate the items to the agents.  The supplier's allocation algorithm has an impact on the requests made by agents, and in effect, on the instance of the allocation problem that must be solved.  Therefore we desire a mechanism that encourages agents to be relatively abstemious, or {\em not too greedy} in choosing their requests, and thereby produces an allocation yielding near-optimal social welfare.  



\subsection{Motivation}

Our model closely resembles a financing competition, where multiple contestants apply for funding provided by one supplier.  Contestants must write an application, or proposal, showing how the resources requested are going to be used to acquire the said value.  The supplier is able to verify the veracity of the proposals and disqualify any contestant that reports a higher valuation than justified.  This is the verification mechanism, and motivates our assumption that agents cannot inflate the reported value.  However, the supplier may not be able to verify that the reported value is the {\em maximum} a contestant could obtain, since it may not know the full capabilities of the contestant.

The verification mechanism can also be thought of as a set of laws or a reputation system.  If an agent obtains the requested items and does not bring the reported value, the repercussions may outweigh any immediate gains.  This understanding of the verification mechanism also justifies the assumption that agents cannot inflate their reported valuation. 

\subsection{Related work}

We show how related literature fits in our setting, categorizing it along two orthogonal dimensions: the power of the verification mechanism and communication complexity, or metaphorically, soundness and completeness.  Fig. \ref{pic:RelatedWork} classifies existing work within these dimensions.  

\begin{figure}[t!]

\center
\begin{tikzpicture}
\draw [very thin] (.3, 0)  -- (9.5, 0) ;
\draw [thin] (3.3, -1.3)--(3.3, 1.3);
\draw [thin] (6.7, -1.3)--(6.7, 1.3);

\draw [->, thick] (0, 1.7)  -- (10, 1.7) ;
\draw (1.5, 1.7) node[below]  {\footnotesize \em none}
	 (5, 1.7) node[below]  {\footnotesize \em partial}
	 (5, 1.7) node[above]  {\footnotesize verification}
	 (8.5, 1.7) node[below]  {\footnotesize \em full};
	
\draw [->, thick] (-.3, -1.5) -- (-.25, 1.5);
\draw
	(-.5, 0) node[rotate = 90] {\footnotesize revelation}  
	(0, .7) node[rotate = 90] {\footnotesize \em full }
	(0, -.7) node[rotate = 90] {\footnotesize \em partial };
	
\draw (8.5, .9) node {\footnotesize knapsack problem};
	
\draw (1.6, .9) node {\footnotesize multi-unit auctions};

\draw (1.6, -.5) node{\footnotesize bounded};
\draw (1.6, -.9) node{\footnotesize communication};

\draw (5, .9) node {\footnotesize mechanisms with};
\draw (5, .5) node {\footnotesize verification};

\draw (5, -.5) node {\footnotesize \textbf {$k$-round HRG}};
\draw (5, -.9) node {\footnotesize \textbf {PoA $1+1/k$}};

\draw (8.5, -.5) node {\footnotesize marginal greedy};
\draw (8.5, -.9) node {\footnotesize PoA $1$};

\end{tikzpicture} \vskip -.15cm
\caption{Problem settings.} \label{pic:RelatedWork} 
\end{figure}
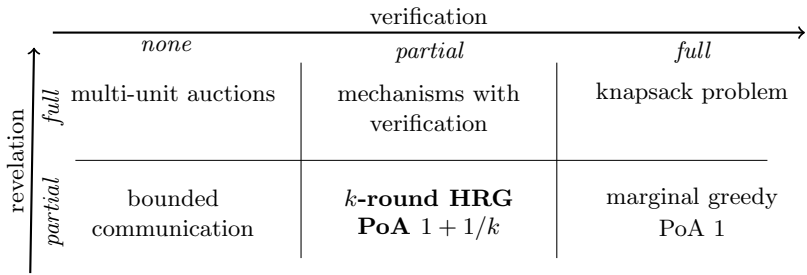
\paragraph{No verification, full revelation.}
This is the most common assumption in the algorithmic mechanism design literature.  Multi-unit auctions model the situation where a verification mechanism does not exist and thus agents must be assumed dishonest.  
Truthfulness can be achieved through VCG payments, but doing so depends on solving the allocation problem optimally, which may be intractable.  Starting with the work of Nisan and Ronen \cite{Nisan01}, the field of algorithmic mechanism design has sought to reconcile selfishness with computational complexity.
Multi-unit auctions have been studied extensively in this context, including truthful mechanisms for single-minded bidders \cite{Mu'alem02,Briest05}, and $k$-minded bidders \cite{Lavi05,Dobzinski07,Dobzinski10}.

More recently Procaccia and Tennenholtz (\cite{Procaccia09}), initiated the study of strategy proof mechanisms without money, which was followed by the adaptation of many previously studied mechanism design problems to the non-monetary setting ( \cite{Alon11}, \cite{Ashlagi10}, \cite{Chen11}, \cite{Dughmi10}, \cite{Guo10}, \cite{Lu10}). 
\paragraph{No verification, partial revelation.}
The multi item allocation problem has also been studied in the setting where dishonest agents only partially reveal their valuation functions.  The main question in this setting concerns the extent to which limiting communication complexity affects mechanism efficiency.  In \cite{Blumrosen02,Blumrosen03}, for example, bid sizes in a single-item auction are restricted to real numbers expressed by $k$ bits.  In \cite{Cohen02}, agent valuation functions are only partially revealed because full revelation would require exponential space in the number of items.

\paragraph{Partial verification, full revelation.}  
Mechanisms with verification have been introduced in \cite{Nisan01} for selfish settings of the task scheduling problem.  The authors show that truthful mechanisms exist for this problem when the mechanism can detect some of the lies, which is very natural in this setting.  More recently, this results were generalized to mechanisms that are collusion resistant (\cite{Penna08}), and to more general optimization functions (\cite{Auletta06}, \cite{Auletta09}, \cite{Ferrante09}), as well as multi parameter agents \cite{Ventre06}.  

\paragraph{Full verification, full revelation.}
If the verification mechanism has full power to ensure agents' honesty and agents must report their full valuation functions, the supplier has complete information and selfishness on the part of the recipients is irrelevant. This setting can be modeled as a multiple-choice knapsack problem solvable by FPTAS \cite{Knapsack04}.

\subsection{Contributions}

This paper extends the study of mechanisms with partial verification to multi unit resource allocation.  Unlike the problems analyzed before, there are polynomial time truthful mechanisms for multi unit auctions.  However, these mechanisms require both full revelation of the agent type, which may be hard to compute and communicate, and currency transfer, which may be impractical in some scenarios.  Our work uses the added power of verification to provide an efficient approximation mechanism for scenarios where currency transfer cannot be modeled.

We propose the highest-ratio greedy (HRG) mechanism for the Funding Game, which provides a Bayesian $PoA$ of 2 under the assumption that valuation functions give diminishing marginal returns (Theorem \ref{thm:singlepoa}).  We also provide an algorithm that computes the Nash equilibrium strategy profile in $O(n^2 \log^2 m)$ time and a best response protocol that converges to a Nash equilibrium profile.  We show that an extension of HRG to multiple rounds can arbitrarily strengthen the pure $PoA$.  In this extension, the supplier partitions the $m$ items into $k$ carefully-sized subsets, and allocates them successively over $k$ consecutive Funding Games.  We show that this mechanism has a pure $PoA$ of $1 + \frac{1}{k}$, yielding a graceful tradeoff between communication complexity and the social welfare (Theorem \ref{thm:multipoa}).   



\section{Preliminaries}

A single-round Funding Game is specified by a set of agents or {\em players} $\{1, ..., n\}$, a set of $m$ identical resources or items, and for each agent $i$ a valuation function $v_i:\{0, ..., m\} \rightarrow \mathbb{R}_{0}^+$ denoting the value $i$ derives from receiving different numbers of items.  We assume all valuation functions satisfy $v_i(0) = 0$, are nondecreasing, and exhibit diminishing marginal returns: $$v_i(x) - v_i(x-1) \geq v_i(x+1) - v_i(x)$$

A strategy or {\em request} of agent $i$ is a pair $s_i(x_i) = (x_i, \tilde{v}_i(x_i))$ specifying the number $x_i$ of items requested, and its valuation for these items.  A request is valid if $\tilde{v}_i(x) \leq v_i(x)$.

A {\it strategy profile} is an $n$-tuple of strategies $\mathbf{s}=(s_1(x_1), ..., s_n(x_n))$.  We denote by $\mathcal{S}_i$ the set of valid strategies for agent $i$, and by $\mathcal{S} = \mathcal{S}_1 \times ... \times \mathcal{S}_n$ the set of valid strategy profiles.  We denote by $X = (X_1, ..., X_n)$ an allocation of the items to the players where $X_i$ is the number of items allocated to player $i$.  Let $\mathcal{X}$ be the set of all valid allocations, i.e. all $X$ such that $\sum_{i\in[n]} X_i \leq m$.
A {\it mechanism}  $M: \mathcal{S} \rightarrow \mathcal{X}$ is an allocation algorithm that takes as input a strategy profile $\mathbf{s}$ and outputs an allocation $X$ of the items to the players.   We will denote by $X^M(\mathbf{s})=(X^M_1(\mathbf{s}), ..., X^M_n(\mathbf{s}))$ the output of mechanism $M$ for strategy profile $\mathbf{s}$.  The {\em payoff} of player $i$ with valuation $v_i$ is its valuation for the number of items it has been allocated: $u_i^M(v_i;\mathbf{s}) =  {v_i(X^M_i(\mathbf{s}))}$.   
If $\mathbf{v}=(v_1, ..., v_n)$ is a valuation function profile we denote by $OPT^\mathbf{v}$ an {\em optimal allocation}, by $sw(OPT^\mathbf{v})$ the social welfare of the optimal allocation, and by $sw^M(\mathbf{v}; \mathbf{s})=\sum_{i \in [n]} u^M_i(v_i; \mathbf{s})$ the social welfare of strategy profile $\mathbf{s}$.  We use $(s_i', \mathbf{s}_{-i})$ to denote the strategy profile $\mathbf{s}$ in which player $i^{th}$ strategy has been replaced by $s_i'$.

A strategy profile $\mathbf{s}$ is a {\it Nash equilibrium} for a Funding Game with valuation functions $\mathbf{v}$ if for any $i$ and any $s'_i \in \mathcal{S}_i$, $u_i^M(v_i; \mathbf{s}) \geq u_i^M(v_i; s'_i, \mathbf{s}_{-i})$.  The {\it Price of Anarchy} (PoA) bounds the ratio of the optimal social welfare and the social welfare of the worst Nash equilibrium in any Funding Game:
$$
PoA^M = \sup_{\mathbf{v},\text{~NE~} \mathbf{s}} \frac{sw(OPT^\mathbf{v})}{sw^M(\mathbf{v}; \mathbf{s})}
$$

In incomplete information games we assume that player $i$'s valuation function $v_i$ is drawn from a set $V_i$ of possible valuation functions, according to some distribution $D_i$.  We denote by $D = D_1 \times ... \times D_n$ the product distribution of all players' valuation functions.  A strategy $\sigma_i$ in an incomplete information game is a mapping $\sigma_i : V_i \rightarrow S_i$ from the set of the possible valuation functions to the set of valid requests.    
Assuming that the distribution $D$ is commonly known, the {\em Bayesian Nash equilibrium} is a tuple of strategies $\sigma = (\sigma_1, ..., \sigma_n)$ such that, for any player $i$, any valuation function $v_i \in V_i$ and any alternate pure strategy $s_i'$: 
$$\mathbb{E}_{v_{-i} \sim D_{-i}}[u^M_i(v_i; \sigma_i(v_i), \sigma_{-i}(\mathbf{v}_{-i})] \geq \mathbb{E}_{v_{-i} \sim D_{-i}} [u^M_i(v_i; s_i', \sigma_{-i}(\mathbf{v}_{-i})]$$
The {\em Bayesian Price of Anarchy} is defined as the ratio between the expected optimal social welfare and that of the worst bayesian Nash equilibrium: 
$$BPoA =  \sup_{D, \text{~BNE~} \sigma} \frac{\mathbb{E}_{\mathbf{v} \sim D} [sw^M(\mathbf{v}; OPT^\mathbf{v}]}{\mathbb{E}_{\mathbf{v} \sim D} [sw^M(\mathbf{v}; \sigma(\mathbf{v})]}$$

\section{Single-round games}
We first observe that the mechanism that solves the induced integer knapsack problem optimally has an unbounded $PoA$.  This can be shown by the following simple example.  Assume that $n$ items are to be allocated to $n$ players with valuation functions $v_i(x) = 1 + x * \epsilon $ for all $i$ and $x > 0$.  A Nash equilibrium of this game is when all players request all items.  The mechanism allocates all items to one player resulting in a social welfare of $1 + n\cdot \epsilon$.  The optimal allocation will allocate one item to each player for a social welfare of $n$. 

For the remainder of this section we analyze the performance of a simple greedy mechanism in a single shot game.  The {\it Highest Ratio Greedy (HRG)} mechanism grants the requests in descending order according to the ratio $v_i(x_i) / x_i$,  breaking ties in the favor of the player with lower index.  If there are not enough items available to satisfy a request completely, the request is satisfied partially.  This is exactly the greedy algorithm for the fractional knapsack problem.
In this section we show that both the pure and Bayesian $PoA$ are 2.  An interesting open problem is whether a mechanism exist for the single round game that improves this $PoA$.  We make use of the notion of smooth games (\cite{Roughgarden12}) which we review below, cast to the Funding Games studied here.  Since we are only considering the Highest Ratio Greedy mechanism we will omit the superscript $M$ from all notations in this section.
\begin{definition}[Smooth game \cite{Roughgarden12}]
	A Funding Game is $(\lambda, \mu)$-smooth with respect to a choice function $c^*: V_1 \times ... \times V_n \rightarrow \mathcal{S}$ and the social welfare objective if, for any valuation function profiles $\mathbf{v}$ and $\mathbf{w}$ and any strategy profile $\mathbf{s}$ that is valid with respect to both $\mathbf{v}$ and $\mathbf{w}$, we have: 
	$$\sum_{i=1}^n u_i(v_i; c_i^*(\mathbf{v}), \mathbf{s}_{-i}) \geq \lambda \cdot sw(\mathbf{v}; c^*(\mathbf{v})) - \mu \cdot sw(\mathbf{w}; \mathbf{s})$$
\end{definition}
The choice function can be thought of as the optimal strategy profile, in our case the strategy profile in which each player requests the number of items received in an optimal allocation, when the valuation function profile is $v$.
\begin{lemma}
Let $OPT^\mathbf{v} = (o^\mathbf{v}_1, ..., o^\mathbf{v}_n)$ be an optimal allocation for valuation profile $\mathbf{v}$ and $O: V_1 \times ... \times V_n \rightarrow \mathcal{S}$ be the optimal strategy choice function, with $O(\mathbf{v}) = ((o^\mathbf{v}_i, v_i(o^\mathbf{v}_i))_{i\in [n]})$.
The Funding Games are  $(1, 1)$-smooth with respect to $O$ and the social welfare objective.
\end{lemma}
\begin{proof} We will use $o_i$ instead of either request $(o^\mathbf{v}_i, v_i(o^\mathbf{v}_i))$ or integer $o^\mathbf{v}_i$.  It will be clear from context whether $o_i$ stands for a request or an integer.  

\noindent Fix valuation function profiles $\mathbf{v}$ and $\mathbf{w}$.  For a strategy profile $\mathbf{s}$ valid with respect to both $\mathbf{v}$ and $\mathbf{w}$ we show that $\sum_{i=1}^n u_i(v_i; o_i, \mathbf{s}_{-i}) \geq sw(\mathbf{v}; O(\mathbf{v})) - sw(\mathbf{w}; \mathbf{s})$.

\noindent Let $A = \{i: u_i(v_i; o_i, \mathbf{s}_{-i}) < u_i(v_i; O(\mathbf{v}))\}$ be the set of players that are allocated more items in the optimal allocation than in profile $(o_i, \mathbf{s}_{-i})$.  It is enough to show that $\sum_{i \in A} u_i(v_i; o_i, \mathbf{s}_{-i}) + sw(\mathbf{w}; \mathbf{s}) \geq \sum_{i \in A} u_i(v_i; O(\mathbf{v}))$.

\noindent For each player $i \in A$, the value per allocated item at profile $(o_i, \mathbf{s}_{-i})$ is at least $\frac{v_i(o_i)}{o_i}$ since by definition $i$ is being allocated less than $o_i$ items, and the valuation functions are concave.  Then,
$u_i(v_i;o_i, \mathbf{s}_{-i}) \geq \frac{v_i(x_i^*)}{x_i^*} \cdot X_i(o_i, \mathbf{s}_{-i})$.
By definition, each player $i\in A$ would be allocated fewer items than $o_i$.  

\noindent Therefore the requests in $ \mathbf{s}_{-i}$ that have a better value per item ratio sum up to $m-X_i(c_i^*(\mathbf{v}), \mathbf{s}_{-i})$ items.  Since the strategy profile $\mathbf{s}$ is assumed to be valid with respect to valuation function profile $\mathbf{w}$, the valuations expressed in $\mathbf{s}$ are at most equal to the valuations $\mathbf{w}$.  We can conclude that for any $i \in A$
$$sw(\mathbf{w}; \mathbf{s}) \geq (m - X_i(o_i, \mathbf{s}_{-i})) \cdot \frac{v_i(o_i)}{o_i}$$

\noindent Then for any $i\in A$,
$u_i(v_i; o_i, \mathbf{s}_{-i}) + sw(\mathbf{w}; \mathbf{s}) \geq m \cdot \frac{v_i(o_i)}{o_i}$.
This is true in particular for player $j\in A$ with the highest value per item ratio $\frac{v_j(o_j)}{o_j}$. Therefore 
\begin{align*}
\sum_{i \in A} u_i(v_i; o_i, \mathbf{s}_{-i})  + sw(\mathbf{w}; \mathbf{s}) 	&\geq u_j(v_j; o_j, \mathbf{s}_{-j}) + sw(\mathbf{w}; \mathbf{s}) \\
									   	&\geq m \cdot \frac{v_j(o_j)}{o_j}\\
										&\geq \sum_{i \in A}u_i(v_i; O(\mathbf{v}))
\end{align*}
which completes the proof. \qed
\end{proof}

\begin{theorem}\label{thm:singlepoa}
Both the pure and Bayesian Price of Anarchy for the Funding Games are equal to 2.
\end{theorem}
\begin{proof}  Since the Funding Games are (1, 1)-smooth with respect to an optimal allocation, the extension theorem in \cite{Roughgarden12} guarantees that the BPoA is bounded by 2.  We now show that the pure PoA is arbitrarily close to 2. 
Consider the Funding Game with $m$ items and two players with valuation functions $v_1(x)  = m$ and $v_2(x) = x$ $\forall x > 0$.  One possible Nash equilibrium strategy is for both players to request all items.  Since the value per item ratios are equal, only the first player will be allocated, for a social welfare of $m$.   The optimal solution allocates one item to the first player and $m-1$ items to the second player for a social welfare of $2m-1$.  Taking $m$ large enough leads to a PoA arbitrarily close to 2.
\qed
\end{proof}

\subsection{Complexity of computing the Nash equilibrium}

We now present an algorithm that finds the Nash equilibrium in the full information setting in $O(n^2 \log^2 m)$ time.  For each player $i$ we use binary search to find the largest request $(\alpha_i, v_i(\alpha_i))$ that passes the isSatisfiable test.  The isSatisfiable function below assures that regardless of the other players requests, there will be at least $\alpha_i$ items available when the request of player $i$ is considered by the greedy algorithm.  It is easy to see that for the resulting strategy profile each player receives exactly as many items as requested and that all items are allocated.  We need to show that if player $i$ increases its request then it will not receive more items.  By the construction of $\alpha_j$, for any player $j\neq i$, player $j$ will receive at least $\alpha_j$ items regardless of the requests of the other players.  Therefore player $i$ cannot receive more than $\alpha_i = m - \sum_{j \neq i} \alpha_j$ by changing its request.

\begin{algorithm}
\caption {isSatisfiable ($i, x_i$)}
\begin{algorithmic}

\FORALL {$j < i$ }
    \STATE $x_j \gets \max \{ x \in [m]:  \frac{v_j(x)}{x} \geq \frac{v_i(\alpha_i)}{\alpha_i}\}$
\ENDFOR

\FORALL {$j > i$ }
    \STATE $x_j \gets \max \{ x \in [m]:  \frac{v_j(x)}{x} > \frac{v_i(\alpha_i)}{\alpha_i}\}$
\ENDFOR

\RETURN \textbf{true} \textbf{if} {$\sum_{j \neq i} x_j \leq m - \alpha_i$} \textbf{else false}
\end{algorithmic}
\end{algorithm}


\section{Multiple-round games}\label{sec:multiround}

In this section we present our main algorithmic result.  We extend the Funding Game introduced in the previous section to multiple rounds, and we show that the $PoA$ of a $k$-round Funding Game is $1+\frac{1}{k}$, yielding a graceful tradeoff between mechanism complexity and the social welfare.  In a $k$-round Funding Game, the supplier partitions the $m$ items into $k$ bundles, which are distributed among the $n$ agents in $k$ successive Funding Games or {\em rounds}.  We assume that the supplier does not reveal the total number of available items $m$, nor the number of rounds $k$ a priori.  In our analysis we assume that the agents play the Nash equilibrium strategy myopically, in each individual round.  This assumption is in line with the maximin principle which states that rational agents will choose a strategy that maximizes their minimum payoff.  If agents never know whether any additional items are going to be awarded in future rounds, they will try to maximize the utility in the current round.  In the Funding Game, this is equivalent to playing the Nash equilibrium strategy.

As above, we use subscripts to indicate player index; we now use superscripts to indicate round index. Let $m^1, ..., m^k$ be the sizes of the bundles awarded in rounds $1, ..., k$ respectively, with $\sum_{t=1}^k m^t= m$.  As before, the agents have valuation functions $v_i:\{0, ..., m\} \rightarrow \mathbb{R}_0^+$, which are normalized ($v_i(0) = 0$), are nondecreasing, and exhibit diminishing marginal returns.

Let $x_i^t$ be the number of items requested by agent $i$ in game $t$ and let $X^t$ be the allocation vector for round $t$.  Let $\alpha_i^t = \sum_{j = 1, .., t} X_i^j$ be the cumulative number of items allocated to agent $i$ in the first $t$ games, with $\alpha_i^0 = 0$ for all $i$.
In round $t$, agent $i$'s valuation function $v_i^t$ is its {\em marginal valuation} given the number of items received in the earlier rounds:
$$v_i^t (x)= v_i(x + \alpha_i^{t-1}) - v_i(\alpha_i^{t-1})$$

Observe that these marginal valuations functions $v_i^t$ are normalized, are nondecreasing and have diminishing marginal returns, just like the full valuation functions $v_i$.  $G^t$ will denote the Funding Game played at round $t$ with $m^t$ items and valuation functions $v_i^t$.  Observe that these individual Funding Games agents are playing at each round depend on how items have been allocated in previous rounds, and indirectly, on players' strategies in previous rounds.

A strategy or request for agent $i$ is a $k$-tuple $s_i(x_i^1, ..., x_i^k) = (s_i^1(x_i^1), ..., s_i^k(x_i^k))$ where $s_i^t(x_i^t) = (x_i^t, v_i^t(x_i^t))$ is the request of player $i$ in game $t$.  We use $s_i$ as a shorthand to denote the strategy of player $i$ in $G$, and $s_i^t$ to denote the strategy of player $i$ in game $t$.  A {\it strategy profile} for a $k$-round Funding Game will refer to an $n$-tuple of strategies $\mathbf{s} = (s_1, ..., s_n)$ and a strategy profile for game  $G^t$ will refer to the $n$-tuple of requests of players in round $t$, $\mathbf{s}^t = (s_1^t, ..., s_n^t)$.
For a strategy profile $\mathbf{s}$, we will write $sw(\mathbf{s}) = \sum_{i=1}^n v_i(\alpha_i^{k})$ for the social welfare of $\mathbf{s}$.   Let $sw(\mathbf{s}^t)$ be the social welfare of $\mathbf{s}^t$.  Let $\Delta^t = \max_i v_i^t(1)$ be the highest marginal value for one item for any agent in round $t$. Observe that $\Delta^t$ is a nonincreasing function of $t$. 

%
%

\begin{definition}
Strategy profile $\mathbf{s}$ is a myopic equilibrium for the $k$-round Funding Game if for each $t$, $\mathbf{s}^t$ is a Nash equilibrium of round $t$. The {\it myopic Price of Anarchy} (PoA) bounds the ratio of the optimal social welfare and the social welfare of the worst myopic equilibrium in any k-round Funding Game:
$$
PoA = \sup_{\mathbf{v} \text{, myopic NE~} \mathbf{s}} \frac{sw(OPT^\mathbf{v})}{sw(\mathbf{s})}
$$
\end{definition}

Our goal is to analyze how a supplier should partition the $m$ items into bundles in order to obtain as good a $PoA$ as possible.  Theorem \ref{thm:multipoa} in this section shows how the $PoA$ relates to the choices of bundle ratios, while in the next section we find the bundle ratios that give the best $PoA$ guarantees.

\begin{lemma}\label{lem:descDelta}
For any myopic Nash equilibrium strategy profile $\mathbf{s}$ for a $k$-round game, we have $\Delta^t \geq  \frac{sw(\mathbf{s}^t)}{m^t} \geq \Delta^{t+1}$ for each $t$.
\end{lemma}
\noindent{\em Proof:}
The first inequality follows from the definition of $\Delta^t$ and the diminishing returns assumption.

For the second inequality, suppose $\Delta^{t+1} > \frac{sw(\mathbf{s}^t)}{m^t}$.  This would imply that either some items are not allocated at $\mathbf{s}^t$ (impossible since $\mathbf{s}^t$ is Nash equilibrium and by assumption $\Delta^{t+1}>0$) or that some winning player $i$ has valuation-per-item ratio $\frac{v_i^t(x_i^t)}{x_i^t} <\Delta^{t+1} = v_j^{t+1}(1)$, for some player $j$. But then $j$ could have successfully requested another item in game $G^t$, meaning $\mathbf{s}^t$ is not Nash equilibrium, and so contradiction.
\qed

\begin{lemma} \label{lem:bound}
For any myopic equilibrium $\mathbf{s}$ of a $k$-round Funding Game, we have:
$$sw(OPT) \leq sw(\mathbf{s}) + \Delta^{k+1} \cdot \sum_{t=1}^k \left(m^t - \frac{sw(\mathbf{s}^t)} {\Delta^t} \right)$$
\end{lemma}

\begin{theorem} \label{thm:multipoa}
Let $y_t = m^t / m^1$.  The $PoA$ of the $k$-round Funding Game with bundle sizes $m^t$ is bounded by:
\begin{equation}\label{eq:poa}
1 + \sup_{x_1, ..., x_n: x_i \geq 1}  \frac{\sum_{t=1}^{k}y_t(1 - \frac{1}{x_t})}{\sum_{t=1}^k y_t \prod_{i=t+1}^k x_i }
\end{equation}
\end{theorem}
\begin{proof}
Let $\mathbf{s}$ be a myopic equilibrium for a $k$-round game. We will show that there exist $x_1, ..., x_k$, $x_i \geq 1$, such that:
$$\frac{sw(OPT)}{sw(\mathbf{s})} \leq  \frac{\sum_{t=1}^{k}y_t(1 - \frac{1}{x_t})}{\sum_{t=1}^k y_t \prod_{i=t+1}^k x_i }  $$
From Lemma \ref{lem:bound}, we have:
$$sw(OPT) \leq sw(\mathbf{s}) + \Delta^{k+1} \cdot \sum_{t=1}^{k} \left( m^t - \frac{sw(\mathbf{s}^t)}{\Delta^t} \right)$$

\noindent Let $x_t = \frac{\Delta^{t}}{\Delta^{t+1}} $, which is at least 1 for each $t$.  Since $\mathbf{s}^t$ is a Nash equilibrium for round $t$, $\Delta^{t+1} \leq \frac{sw(\mathbf{s}^t)}{m^t} $ for each $t$.

\vspace{.1in}
\noindent Then we have:
\begin{align}
                sw(OPT) - sw(\mathbf{s})     &\leq\Delta^{k+1} \cdot \sum_{t=1}^{k}  \left( m^t - \frac{sw(\mathbf{s}^t)}{\Delta^t} \right) \nonumber \\
                                &\leq \Delta^{k+1} \cdot \sum_{t=1}^{k}  m^t \left(1 - \frac{\Delta^{t+1}}{\Delta^t} \right) \nonumber\\
                                &\leq  m^1 \Delta^{k+1} \cdot \sum_{t=1}^{k}  y_t \left( 1 - \frac{1}{x_t} \right) \label{ineq1}
\end{align}

\noindent  Observe that $\Delta^{t} = \Delta^{k+1}\prod_{i=t}^{k} x_i$.  Therefore:

\begin{align}
sw(\mathbf{s}) = \sum_{t=1}^k sw(\mathbf{s}^t)    &\geq \sum_{t=1}^k m^t \Delta^{t+1}  
                    \geq m^1 \Delta^{k+1} \cdot \sum_{t=1}^k y_t \cdot \prod_{i=t+1}^k x_i  \label{ineq2}
\end{align}

\noindent From (\ref{ineq1}) and (\ref{ineq2}) it follows that for any $k$-round game with bundle sizes $m^t$, there exist $x_1, ..., x_k$ such that:
\begin{align*}
        PoA &= 1 + \sup \frac{sw(OPT) - sw(\mathbf{s})}{sw(\mathbf{s})}    
        \leq 1+ \sup_{x_t \geq 1}\frac{\sum_{t=1}^{k}  y_t \left( 1 - \frac{1}{x_t} \right)}{\sum_{t=1}^k y_t \cdot \prod_{i=t+1}^k x_i}
\end{align*}
\qed
\end{proof}


\section{Evaluating the PoA}

In this section we present two results analyzing the expression (\ref{eq:poa}) above. 
Theorem \ref{theorem:opt1} shows that supremum of this expression taken over all valid choices of $x_t$ but  fixing $y_t=t$ is $1/k$.  This corresponds to bundle sizes $m_1, 2\cdot m_1, ..., k \cdot m_1$ for some $m_1$, 
indicating that the PoA {\em for such bundle sizes} equals $1+1/k$. 

Second, we show that the min-sup of this expression, now also taken over choices of $y_i$, which corresponds to considering all possible choices of bundle sizes, equals the same value $1/k$, indicating that there is no better partition of the items.

\begin{theorem} \label{theorem:opt1}
Let
\begin{equation*}
F(x_1, ..., x_k) = \frac{\sum\limits_{i=1}^{k}i(1 - \frac{1}{x_i})}{\sum\limits_{i=1}^k i \prod\limits_{j=i+1}^k x_j } \quad x_i \geq 1, ~i=1,...,k
\end{equation*}

Then $\sup\limits_{\mathbf{x}} F(\mathbf{x}) = \frac{1}{k}$.
\end{theorem}

\noindent\begin{proof}
First observe that:
\begin{equation*}
F(\mathbf{x}) = \frac{1 - \frac{1}{x_1} + \sum\limits_{i=2}^{k}i(1 - \frac{1}{x_i})}{\sum\limits_{i=1}^k i \prod\limits_{j=i+1}^k x_j } < \lim\limits_{x_1\to\infty}  F(\mathbf{x})
\end{equation*}

\noindent If we set $x_i = \frac{i}{i-1}$, $i=2,...,k$, we have $ \lim\limits_{x_1\to\infty}  F(\mathbf{x})  =  \frac {1}{k}$.
It remains to show that $\lim\limits_{x_1\to\infty}  F(\mathbf{x}) \leq \frac{1}{k}$.
We note that the following inequalities are equivalent:
\begin{align}
\lim\limits_{x_1\to\infty}  F(\mathbf{x})  \leq \frac{1}{k} \Leftrightarrow \nonumber 
& \lim\limits_{x_1\to\infty} \left(\sum\limits_{i=1}^{k} i \prod\limits_{j=i+1}^{k}x_j - k \sum\limits_{i=1}^{k} i(1 -\frac{1}{x_i}) \right) \geq 0  \Leftrightarrow  \nonumber\\
& \lim\limits_{x_1\to\infty} \left(\sum\limits_{i=1}^{k} (i z_i + ik \cdot \frac{z_i}{z_{i-1}}) - \sum\limits_{i=1}^{k} i k\right)  \geq 0, \label{eq:transform}\\
& \text{where } z_i = \prod\limits_{j=i+1}^{k} x_j,~ i= 1,...,k-1; z_k = 1; z_0 = x_1z_1 \nonumber
\end{align}

\noindent Now define a function $C:[0, \infty)^{k-1} \rightarrow \mathbb{R}$, 
$C(\mathbf{z}) = \sum\limits_{i=1}^{k} (i z_i + ik \cdot \frac{z_i}{z_{i-1}}) - \sum\limits_{i=1}^{k} i k$.
Notice that $C$ is a function of $k-1$ variables since $z_0$ and $z_k$ are fixed.  Also notice that the domain of $C$ strictly includes the domain of $\mathbf{z}$ as defined in Eq. (\ref{eq:transform}).
To complete the proof, we show that $C(\mathbf{z}) \geq 0$ for any $\mathbf{z}\in [0, \infty)^{k-1}$.  We will do this in two steps:  (i) showing that $C(\mathbf{z})$ has a unique stationary point, and then (ii) showing that $C(\mathbf{z}) \geq 0$ at any of the domain boundaries and the stationary point.

\noindent {\bf $\mathbf{C(z)}$ has a unique stationary point.}
Let $\mathbf{a} = (a_1, ..., a_{k-1})$ be a stationary point for function $C$, and let $a_0 = z_0 = x_1z_1$ and $a_k = z_k = 1$:
\begin{equation}
    \frac{\partial C}{\partial z_i} (\mathbf{a}) =
        i + \frac{i k}{a_{i-1}} - k (i+1) \frac{a_{i+1}}{a_i^2} = 0, ~ i=1,...,k-1 \label{eq:partialDerivative}
\end{equation}

\noindent We show now by induction that each $a_i$ can be written as a function of $a_1$.
For the base case, let $a_0 = x_1 \cdot a_1 = f_0(a_1)$ and $f_1(a_1) = a_1$.

\noindent Now assume that $a_{i-1} = f_{i-1}(a_1)$ and $a_i = f_i(a_1)$.  Then we will define $a_{i+1}$ as a function of $a_1$ as follows. From Eq. (\ref{eq:partialDerivative}) we can infer:
\begin{align}
   a_{i+1} &= \left(i + \frac{ik}{a_{i-1}}\right) \cdot \frac{a_{i}^2}{k (i+1)} \nonumber \\
       a_{i+1} &= \left(i + \frac{ik}{f_{i-1}(a_1)}\right) \cdot \frac{f_{i}^2(a_1)}{k(i+1)}   \triangleq f_{i+1}(a_1) \label{eq:stationary}
\end{align}
where $f_{i+1}(\cdot)$ is the name given to the expression in Eq. (\ref{eq:stationary}) as a function of $a_1$.

\noindent Therefore the equations $a_i = f_i(a_1)$, $i=1, ..., k-1$ uniquely define a stationary point $\mathbf{a}$ with respect to $a_1$.  To show that the stationary point $\mathbf{a}$ is unique, we only need to show that $f_k(a_1) = 1$ has a unique solution.  For this it is sufficient to show that the derivative of $f_k$ with respect to $a_1$ is always positive: $f'_k(a_1) > 0$.

\noindent We show this by induction on $i=0,...,k$. Let $h_i = \frac{f_{i}}{f_{i-1}}$, $i =2, ..., k-1 $.   The inductive hypothesis is that $f'_i(a_1) > 0$, $i=1,...,k$ and $h_j(a_1) > 0$ and $h'_j(a_1) > 0$, $j=2,...,k$.

\noindent For the base case, observe the following:
\begin{align}
f_1(a_1) & =  a_1 > 0  \text{ and }  f'_1(a_1) =  1> 0 \nonumber\\
f_2(a_1) & = \frac{x_1a_1^2 + ka_1}{2kx_1}   \text{ and }   f'_2(a_1) = \frac{2x_1a_1 + k}{2kx_1} > 0 \nonumber\\
h_2(a_1) & =  \frac{f_2(a_1)}{f_1(a_1)} = \frac{x_1a_1+k}{2kx_1} > 0
\text{ and }  h'_2(a_1) = \frac{x_1}{2kx_1}> 0  \nonumber
\end{align}

\noindent Now assume that $f'_{i}(a_1) > 0$, $h_{i}(a_1)>0$, and $h'_{i}(a_1)>0$.  We then observe that $f'_{i+1}(a_1)$, $h_{i+1}(a_1)$ and $h'_{i+1}(a_1)$ are all strictly positive:
\begin{align}
    f'_{i+1}(a_1)  &= h'_i(a_1) \cdot f_i(a_1) + h_i(a_1) \cdot f'_i(a_1) >0 \nonumber\\
    h_{i+1}(a_1) &= \left( i + \frac{ik}{f_{i-1}(a_1)} \right) \cdot  \frac{f_i(a_1)}{k
    (i+1)} \nonumber\\
            & = \frac{i}{k (i+1)} f_i(a_1) + \frac{i}{i+1} \cdot h_i(a_1) >0 \nonumber\\\
    h'_{i+1}(a_1) &= \frac{i}{k (i+1)} f'_i(a_1) + \frac{i}{i+1} \cdot h'_i(a_1) >
    0 \nonumber\
\end{align}

\noindent This shows that the equation $f_k(a_1) = 1$ has a unique solution, and thus concludes step (i).

\noindent {\bf $\mathbf{C(z)} \geq 0$ at all boundary points and at the unique stationary point.}
First observe that $a_i = \frac {k}{i}$  satisfies Eq. (\ref{eq:stationary}), $i=1,...,k$ and hence $\mathbf{a}=(a_1, ..., a_{k-1})$ is the unique stationary point for $C$.  Now we show that $C(\mathbf{a})\geq0$:
\begin{align}
C(\mathbf{a})    =\sum\limits_{i=1}^{k} (i a_i + ik \cdot
\frac{a_i}{a_{i-1}}) - \sum\limits_{i=1}^{k} i k \nonumber
             =\sum\limits_{i=1}^{k} (k + k (i-1)) - \sum\limits_{i=1}^{k} i k = 0 \nonumber
\end{align}

Let $\mathbf{b} = (b_1, ..., b_{k-1})$ be a boundary point.  Then we must show that $C(\mathbf{b}) \ge 0$.   Since $\mathbf{b}$ is a boundary point there must exist $j$ such that $b_j = 0$ or $b_j = \infty$:
\begin{equation*}
C(\mathbf{b}) = \sum\limits_{i=1}^{k} (i b_i + ik \cdot \frac{b_i}{b_{i-1}}) - \sum\limits_{i=1}^{k} i k
\end{equation*}

The only negative term is $ \sum_{i=1}^{k} i k $, which is constant with respect to $\mathbf{b}$.  If $b_j = 0$ for some $j$, then the positive term $(i+1)k \cdot \frac{b_{i+1}}{b_i}$ is infinite and $C(\mathbf{b}) > 0$.  On the other hand, if $b_j = \infty$ for some $j$, then the positive term $ik \cdot \frac{b_{i}}{b_{i+1}}$ is infinite and again $C(\mathbf{b}) > 0$.
Steps (i) and (ii) above show that $C(\mathbf{z}) \geq 0$ $\forall z \in [0, \infty)^{k-1}$ and therefore $C(\mathbf{z}) \geq 0$ on the restricted domain of equation (\ref{eq:transform}), which completes the proof. \qed
\end{proof}

\begin{corollary}\label{cor:PoA}
The $PoA$ for the k-round Funding Games with bundle ratios $\frac{m_t}{m_1} = t$ is $1+\frac{1}{k}$.
\end{corollary}

\begin{theorem}\label{theorem:opt2}
Let
\begin{equation}
G(\mathbf{x}, \mathbf{y}) = \frac{\sum\limits_{i=1}^{k} y_i (1 -
\frac{1}{x_i})}{\sum\limits_{i=1}^{k} y_i
\prod\limits_{j=i+1}^{k}x_j} \quad y_i \geq 0; \;x_i \geq 1,
~i=1,...,k \nonumber
\end{equation}

Then $\min\limits_{\mathbf{y}}\sup\limits_{\mathbf{x}} G(\mathbf{x}, \mathbf{y}) = \frac 1k$.

\end{theorem}


\section{Discussion}

In this paper, we introduced the Funding Game, a novel formulation of resource allocation for agents whose valuation declarations can be verified, but reveal only partial information.  We analyzed the $PoA$ for the pure and Bayesian Nash equilibrium and showed that allocating the resources in multiple successive rounds can improve the pure $PoA$ arbitrarily close to 1.  There are two directions in which this work can be extended.  First, our mechanism relies on the assumption that the valuation functions are concave.  An interesting open problem is finding an efficient mechanism for general valuation functions.  Second, it might be desirable to develop efficient verification mechanisms for combinatorial settings, where players' valuation functions are defined on subsets of items.  

\renewcommand{\bibname}{References}
\bibliographystyle{splncs03}

\vspace{1cm}
\begingroup
\let\clearpage\relax

\bibliography{bibliography.bib}
\endgroup

\end{document}